\documentclass[conference]{IEEEtran}
\IEEEoverridecommandlockouts
\PassOptionsToPackage{bookmarks=false}{hyperref}
\usepackage{amsthm}
\usepackage{amssymb}
\usepackage{adjustbox}
\usepackage[ruled,linesnumbered,vlined]{algorithm2e}
\usepackage{balance}  % for  \balance command ON LAST PAGE  (only there!)
\usepackage{dsfont}
\usepackage{graphicx}
\usepackage{hyperref}
\usepackage{mathtools}
\usepackage{multirow}
\usepackage{subcaption}
\usepackage{xcolor}

\newtheorem{theorem}{Theorem}

\newtheorem{definition}{Definition}

\newcommand{\kb}{\ensuremath{\mathcal{K}}}
\newcommand{\entset}{\ensuremath{U}}
\newcommand{\litset}{\ensuremath{L}}
\newcommand{\attrset}{\ensuremath{A}}
\newcommand{\relset}{\ensuremath{R}}
\newcommand{\tripleset}{\ensuremath{T}}
\newcommand{\atset}{\ensuremath{T_{attr}}}
\newcommand{\rtset}{\ensuremath{T_{rel}}}

\newcommand{\pair}{\ensuremath{p}}
\newcommand{\match}{\ensuremath{m}}
\newcommand{\pmatch}{\ensuremath{\match_\pair}}
\newcommand{\ent}{\ensuremath{u}}
\newcommand{\rel}{\ensuremath{r}}
\newcommand{\attr}{\ensuremath{a}}
\newcommand{\lit}{\ensuremath{l}}

\newcommand{\ergraph}{\ensuremath{\mathcal{G}}}
\newcommand{\vertex}{\ensuremath{v}}
\newcommand{\vertexset}{\ensuremath{V}}
\newcommand{\vmatch}{\ensuremath{\match_v}}
\newcommand{\vvmatch}{\ensuremath{\match_{v'}}}
\newcommand{\edgeset}{\ensuremath{E}}
\newcommand{\unresolved}{\ensuremath{C}}
\newcommand{\pergraph}{\ensuremath{\mathcal{F}}}

\newcommand{\benefit}{\ensuremath{\mathtt{benefit}}}
\newcommand{\inferred}{\ensuremath{\mathtt{inferred}}}
\newcommand{\question}{\ensuremath{q}}
\newcommand{\qmatch}{\ensuremath{\match_\question}}
\newcommand{\questionset}{\ensuremath{Q}}
\newcommand{\answerset}{\ensuremath{H}}
\newcommand{\qlimit}{\ensuremath{\mu}}

\newcommand{\initialmatches}{\ensuremath{M_{in}}}
\newcommand{\candidatematches}{\ensuremath{M_c}}
\newcommand{\attrmatches}{\ensuremath{M_{at}}}
\newcommand{\partition}{\ensuremath{B}}

\newcommand{\consistency}{\ensuremath{\epsilon}}
\def\valueset_#1{N_{\ent_{#1}}^{\attr_{#1}}}
\def\neighborset_#1{N_{\ent_{#1}}^{\rel_{#1}}}
\def\neighbormatch{M_{u_1,u_2}}
\def\neighbormatchcnt{L_{u_1,u_2}}

\newcommand{\qdelta}{\Delta q}

\newcommand{\worker}{\ensuremath{w}}
\newcommand{\trueworkers}{\ensuremath{W_T}}
\newcommand{\falseworkers}{\ensuremath{W_F}}

\newcommand{\pairpartition}{\ensuremath{N_\pair}}

\begin{document}

% ****************** TITLE ****************************************

\title{Crowdsourced Collective Entity Resolution with Relational Match Propagation}

\author{\IEEEauthorblockN{Jiacheng Huang\IEEEauthorrefmark{2},
Wei Hu\thanks{\IEEEauthorrefmark{1}Corresponding author}\IEEEauthorrefmark{2}\IEEEauthorrefmark{1}, Zhifeng Bao\IEEEauthorrefmark{3} and
Yuzhong Qu\IEEEauthorrefmark{2}}
\IEEEauthorblockA{\IEEEauthorrefmark{2}{State Key Laboratory for Novel Software Technology, Nanjing University, Nanjing, China}\\
\IEEEauthorrefmark{3}{RMIT University, Melbourne, Australia}\\
Email: jchuang.nju@gmail.com, whu@nju.edu.cn, zhifeng.bao@rmit.edu.au, yzqu@nju.edu.cn}
}

\date{\today}

\maketitle

\begin{abstract}
Knowledge bases (KBs) store rich yet heterogeneous entities and facts. Entity resolution (ER) aims to identify entities in KBs which refer to the same real-world object. Recent studies have shown significant benefits of involving humans in the loop of ER. They often resolve entities with pairwise similarity measures over attribute values and resort to the crowds to label uncertain ones. However, existing methods still suffer from high labor costs and insufficient labeling to some extent. In this paper, we propose a novel approach called crowdsourced collective ER, which leverages the relationships between entities to infer matches jointly rather than independently. Specifically, it iteratively asks human workers to label picked entity pairs and propagates the labeling information to their neighbors in distance. During this process, we address the problems of candidate entity pruning, probabilistic propagation, optimal question selection and error-tolerant truth inference. Our experiments on real-world datasets demonstrate that, compared with state-of-the-art methods, our approach achieves superior accuracy with much less labeling.
\end{abstract}

\section{Introduction}
\label{sec:intro}

Knowledge bases (KBs) store rich yet heterogeneous entities and facts about the real world, where each fact is structured as a triple in the form of $(entity, property, value)$. Entity resolution (ER) aims at identifying entities referring to the same real-world object, which is critical in cleansing and integration of KBs. Existing approaches exploit diversified features of KBs, such as attribute values and entity relationships, see surveys~\cite{getoor2012,elmagarmid2007survey,sun2019progress,bleiholder2009data}. Recent studies have demonstrated that \emph{crowdsourced ER}, which recruits human workers to solve micro-tasks (e.g., judging if a pair of entities is a match), can improve the overall accuracy.

Current crowdsourced ER approaches mainly leverage \emph{transitivity}~\cite{wang2013,whang2013,firmani2016} or \emph{monotonicity}~\cite{arasu2010,gokhale2014corleone,das2017falcon,qian2017erlearn,zhuang2017hike} as their resolution basis. The transitivity-based approaches rely on the observation that the match relation is usually an \emph{equivalence} relation. The monotonicity-based ones assume that each pair of entities can be represented by a similarity vector of attribute values, and the binary classification function, which judges whether a similarity vector is a match, is monotonic in terms of the \emph{partial order} among the similarity vectors.

However, both kinds of approaches can hardly infer matches across different types of entities. Let us see Figure~\ref{fig:graph} for example. The figure shows a directed graph, called \emph{entity resolution graph} (ER graph), in which each vertex denotes a pair of entities and each edge denotes a relationship between two entity pairs. Assume that $(\textit{y:Joan},\textit{d:Joan})$ is labeled as a match, the birth place pair $(\textit{y:NYC},\textit{d:NYC})$ is expected to be a match. Since these two pairs are in different equivalence classes, the transitivity-based approaches are apparently unable to take effect. As different relationships (like \textit{y:directedBy} and \textit{y:wasBornIn}) make most similarity vectors of entities of different types incomparable, the monotonicity-based approaches have to handle them separately.

\begin{figure}
\centering
\includegraphics[width=\columnwidth]{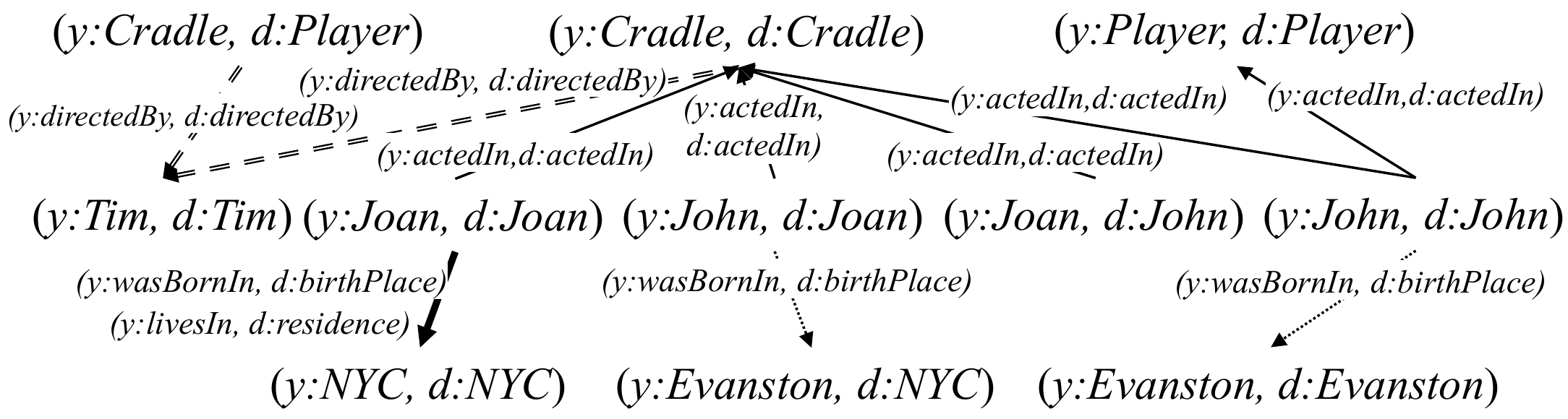}
\caption{An ER graph example between YAGO and DBpedia. Entities in YAGO are prefixed by ``\textit{y:}'', and entities in DBpedia are prefixed by ``\textit{d:}''. \textit{Joan}, \textit{John} and \textit{Tim} are persons. \textit{Cradle} and \textit{Player} are movies. \textit{NYC} and \textit{Evanston} are cities.}
\label{fig:graph}
\end{figure}

In this paper, we propose a new approach called Remp (\underline{Re}lational \underline{m}atch \underline{p}ropagation) to address the above problems. The main idea is to leverage \emph{collective ER} that resolves entities connected by relationships jointly and distantly, based on a small amount of labels provided by workers. Specifically, Remp iteratively asks workers to label a few entity pairs and propagates the labeling information to their neighboring entity pairs in distance, which are then resolved jointly rather than independently. There remain two challenges to achieve such a crowdsourced collective ER.

The first challenge is how to conduct an effective relational match propagation. Relationships like functional/inverse functional properties in OWL \cite{schneider2004owl} (e.g., \textit{y:wasBornIn}) provide a strong evidence, but these properties only account for a small portion while the majority of relationships is multi-valued (e.g., \textit{actedIn}). Multi-valued relationships often connect non-matches to matches (e.g., $(\textit{y:John}, \textit{d:Joan})$ is connected to $(\textit{y:Cradle},\textit{d:Cradle})$ in Figure~\ref{fig:graph}). Therefore, we propose a new relational match propagation model, to decide which neighbors can be safely inferred as matches.

The second challenge is how to select good questions to ask workers. For an ER graph involving two large KBs, the number of vertices (i.e. candidate questions) can be quadratic. We introduce an entity pair pruning algorithm to narrow the search space of questions. Moreover, different questions have different inference power. In order to maximize the expected number of inferred matches, we propose a question selection algorithm, which chooses possible entity matches scattered in different parts of the ER graph to achieve the largest number of inferred matches.

In summary, the main contributions of this paper are listed as follows:
\begin{itemize}
\item We design a partial order based entity pruning algorithm, which significantly reduces the size of an ER graph.
\item We propose a relational match propagation model, which can jointly infer the matches between different types of entities in distance.
\item We formulate the problem of optimal multiple questions selection with cost constraint, and design an efficient algorithm to obtain approximate solutions.
\item We present an error-tolerant method to infer truths from imperfect human labeling. Moreover, we train a classifier to handle isolated entity pairs.
\item We conduct real-world experiments and comparison with state-of-the-art approaches to assess the performance of our approach. The experimental results show that our~approach achieves superior accuracy with much fewer labeling tasks.
\end{itemize}

\noindent\textbf{Paper organization.} Section~\ref{sect:lit} reviews the literature. Section~\ref{sect:def} defines the problem and sketches out the approach. In Sections~\ref{sect:graph}--\ref{sect:inference}, we describe the approach in detail. Section~\ref{sect:exp} reports the experiments and results. Last, Section~\ref{sect:concl} concludes this paper.

\section{Related Work}
\label{sect:lit}

% In this section, we discuss two lines of related work, namely crowdsourced ER and collective ER.

\subsection{Crowdsourced ER}

\noindent\textbf{Inference models.} Based on the transitive relation of entity matches, many approaches such as \cite{wang2013,vesdapunt2014} make use of prior match probabilities to decide the order of questions. Firmani et al.~\cite{firmani2016} proved that the optimal strategy is to ask questions in descending order of entity cluster size. They formulated the problem of crowdsourced ER with early termination and put forward several question ordering strategies. Although the transitive relation can infer matches within each cluster, workers need to check all clusters. 

On the other hand, Arasu et al.~\cite{arasu2010} investigated the monotonicity property among the similarity vectors of entity pairs. Given two similarity thresholds $\mathbf{s}_1$, $\mathbf{s}_2$ and $\mathbf{s}_1 \succeq \mathbf{s}_2$, we have $\mathrm{Pr}[\ent_1\simeq\ent_2 \,|\, \mathbf{s}(\ent_1, \ent_2) \succeq \mathbf{s}_1] \geq \mathrm{Pr}[\ent_1\simeq\ent_2 \,|\, \mathbf{s}(\ent_1, \ent_2) \succeq \mathbf{s}_2]$. ALGPR~\cite{arasu2010} and ERLEARN~\cite{qian2017erlearn} use the monotonicity property to search new thresholds, and estimate the precision of results. In particular, the partial order based approaches~\cite{tao2018, chai2018power,zhuang2017hike} explore similarity thresholds among similarity vectors. Furthermore, POWER \cite{chai2018power} groups similarity vectors to reduce the search space. Corleone~\cite{gokhale2014corleone} and Falcon~\cite{das2017falcon} learn random forest classifiers, where each decision tree is equivalent to a similarity vector. However, these approaches are designed for ER with single entity type. To leverage monotonicity on ER between KBs with complex type information, HIKE~\cite{zhuang2017hike} uses hierarchical agglomerative clustering to partition entities with similar attributes and relationships, and uses the monotonicity techniques on each entity partition to find matches. Although our approach also uses monotonicity, it only uses monotonicity to prune candidate entity pairs. In addition, our approach allows match inference between different entity types (e.g., from persons to locations) to reduce the labeling efforts.

\noindent\textbf{Question interfaces.} Pairwise and multi-item are two common question interfaces. The pairwise interface asks workers to judge whether a pair of entities is a match \cite{firmani2016,verroios2015}. Differently, Marcus et al.~\cite{marcus2011} proposed a multi-item interface to save questions, where each question contains multiple entities to be grouped. Wang et~al. \cite{wang2012} minimized the number of multi-item questions on the given entity pair set such that each question contains at most $k$ entities. Waldo~\cite{verroios2017waldo} is a recent hybrid interface, which optimizes the trade-off between cost and accuracy of the two question interfaces based on task difficulty. The above approaches do not have the inference power and they may generate a large amount of questions.

\noindent\textbf{Quality control.} To deal with errors produced by workers, quality control techniques~\cite{whang2013,verroios2017waldo,galhotra2018} leverage the correlation between matches and workers to find inaccurate labels, and improve the accuracy by asking more questions about uncertain ones. These approaches gain improvement by redundant labeling.

\subsection{Collective ER} 
In addition to attribute values, collective ER~\cite{rastogiL2011mln,bohm2012linda,altowim2014progressive,efthymiou2019minoaner} further takes the relationships between entities into account. CMD~\cite{kimmig2017cmd} extends the probabilistic soft logic to learn rules for ontology matching. LMT~\cite{kouki2017} learns soft logic rules to resolve entities in a familial network. Because learning a probabilistic distribution on large KBs is time-consuming, PARIS~\cite{suchanek2011paris} and SiGMa~\cite{julien2013sigma} implement message passing-style algorithms that obtain seed matches created by hand crafted rules and pass the match messages to their neighbors. However, they do not leverage crowdsourcing to improve the ER accuracy and may encounter the error accumulation problem.

\section{Approach Overview}
\label{sect:def}

\begin{figure*}
\begin{minipage}{.29\textwidth}
	\captionof{table}{Frequently used notations}
	\label{tab:notations}
	\begin{adjustbox}{width=\columnwidth}
	\begin{tabular}{|c|l|}
		\hline 	Notations & Descriptions \\
		\hline	$\mathcal{K}, u$ & a KB and an entity \\
		$r, a$ & a relationship, and an attribute \\
				$N_u^r, N_u^a$ & the value sets of $r$ and $a$ w.r.t. $u$ \\
				$\pair, \question$ & an entity pair, and a question\\
				$\pmatch, \qmatch$ & the event that $\pair$ and $\question$ is a match\\
				% $p, m_p$ & An entity pair, and the event that $p$ is a match \\
				% $q, m_q$ & A question, and the event that $q$ is labeled as a match \\
				$M$ & a set of entity or attribute matches\\
				$C$ & a set of candidate questions \\
				$Q$ & a set of asked questions \\
				$H$ & a set of labels \\
		\hline
	\end{tabular}
	\end{adjustbox}
\end{minipage}
\begin{minipage}{.7\textwidth}
	\vspace{2.3mm}
	\includegraphics[width=\textwidth]{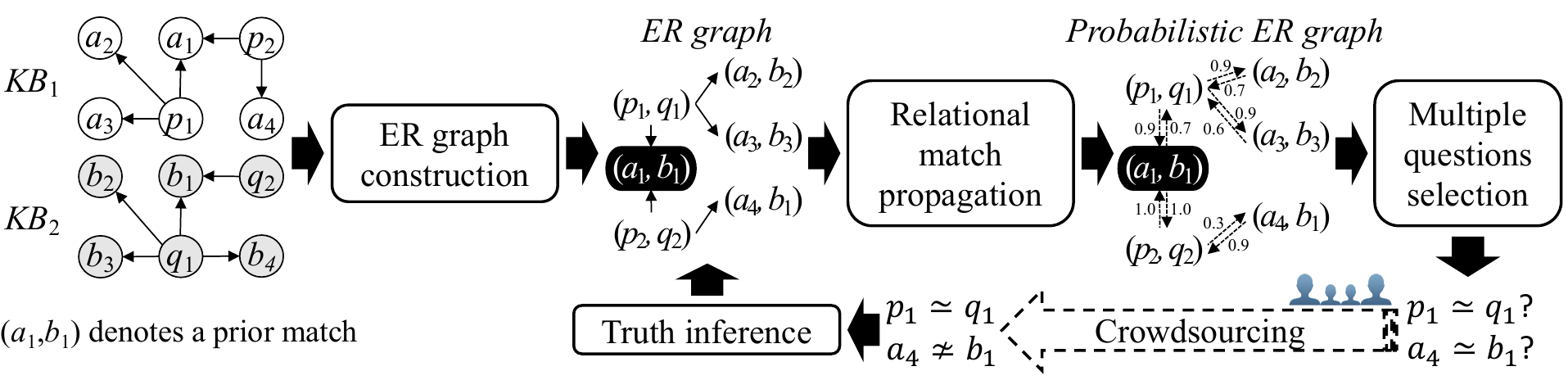}
	\caption{Workflow of the proposed approach}
	\label{fig:workflow}
\end{minipage}
\end{figure*}

In this section, we present necessary preliminaries to define our problem, followed by a general workflow of our approach. Frequently used notations are summarized in Table~\ref{tab:notations}.

\subsection{Preliminaries \& Problem Definition}
\label{subsect:prelim}

\emph{KBs} store rich, structured real-world facts. In a KB, each fact is stated in a triple of the form $(entity, property, value)$, where $property$ can be either an attribute or a relationship, and $value$ can be either a literal or another entity. The sets of entities, literals, attributes, relationships and triples are denoted by $\entset, \litset, \attrset, \relset$ and $\tripleset$, respectively. Therefore, a KB is defined as a 5-tuple $\kb=(\entset,\litset,\attrset,\relset,\tripleset)$. Moreover, attribute triples $\atset\subseteq \entset\times \attrset\times \litset$ attach entities with literals, e.g., $(\textit{Leonardo da Vinci}, \textit{birth date}, \textrm{``1452-4-15"})$, and relationship triples $\rtset\subseteq \entset\times \relset\times \entset$ link entities by relationships, e.g., $(\textit{Leonardo da Vinci}, \textit{works}, \textit{Mona Lisa})$.

\emph{Entity Resolution (ER)} aims to resolve entities in KBs denoting the same real-world thing. Let $\ent_1,\ent_2$ denote two entities in two different KBs. We call the entity pair $\pair=(\ent_1,\ent_2)$ a \emph{match} and denote it by $\ent_1\simeq \ent_2$ or $\pmatch$ if $\ent_1,\ent_2$ refer to the same. In contrast, we call $\pair=(\ent_1,\ent_2)$ a \emph{non-match} and denote it by $\ent_1\not\simeq \ent_2$ if $\ent_1,\ent_2$ refer to two different objects. Both matches and non-matches are regarded as \emph{resolved} entity pairs, and other pairs are regarded as \emph{unresolved}. Traditionally, \emph{reference matches} (i.e., gold standard) are used to evaluate the quality of the ER results, and precision, recall and F1-score are widely-used metrics.

\emph{Crowdsourced ER} carries out ER with human helps. Usually, it executes several human-machine loops, and in each loop, the machine picks one or several questions to ask workers to label them and updates the ER results in terms of the labels. Due to the monetary cost of human labors, a crowdsourced ER algorithm is expected to ask limited questions while obtaining as many results as possible. 

\begin{definition}[Crowdsourced Collective ER] Given two KBs $\kb_1$ and $\kb_2$, and a budget, the crowdsourced collective ER problem is to maximize recall with a precision restriction by asking humans to label tasks while not exceeding the budget.
\end{definition}
Specifically, we assume that both KBs contain ``dense'' relationships and focus on using matches obtained from workers to jointly infer matches with relationships.

%--------------------
\subsection{A Workflow of Our Approach}
\label{subsect:overview}

Given two KBs as input, Figure~\ref{fig:workflow} shows the workflow of our approach to crowdsourced collective ER. After iterating four processing stages, the approach returns a set of matches between the two KBs.

\begin{enumerate}
\item\textbf{ER graph construction} aims to construct a small ER graph by reducing the amount of vertices (i.e. entity pairs). It first conducts a similarity measurement  to filter out some non-matches. At the same time, it uses some matches obtained from exact matching~\cite{zhuang2017hike,julien2013sigma,zhuang2016pba} to calculate the similarities between attributes and find attribute matches. Then, based on the attribute matches, it assembles the similarities between values to similarity vectors, and leverages the natural partial order on the vectors to prune more vertices.

\item\textbf{Relational match propagation} models how to use matches to infer the match probabilities of unresolved entity pairs in each connected component of the ER graph. It first uses some matches and maximum likelihood estimation to measure the consistency of relationships. Then, based on the consistency of relationships and the ER graph structure, it computes the conditional match probabilities of unresolved entity pairs given the matches. The conditional match probabilities derive a probabilistic ER graph.

\item\textbf{Multiple questions selection} selects a set of unresolved entity pairs in the probabilistic ER graph as questions to ask workers. It models the discovery of inferred match set for each question as the all-pairs shortest path problem and uses a graph-based algorithm to solve it. We prove that the multiple questions selection problem is NP-hard and design a greedy algorithm to find the best questions to ask.

\item\textbf{Truth inference} infers matches based on the results labeled by workers. It first computes the posterior match probabilities of the questions based on the quality of the workers, and then leverages these posterior probabilities to update the (probabilistic) ER graph. Also, for isolated entity pairs, it builds a random forest classifier to avoid asking the workers to check them one by one.
\end{enumerate}

The approach stops asking more questions when there is no unresolved entity pair that can be inferred by relational match propagation.

\section{ER Graph Construction}
\label{sect:graph}
\subsection{ER Graph}
Graph structures~\cite{dong2005reference,papenbrock2015} are widely used to model the resolution states of entity pairs and the relationships between them. For example, Dong et al.~\cite{dong2005reference} proposed dependency graph to model the dependency between similarities of entity pairs. In this paper, we use the notion of \emph{ER graph} to denote this graph structure. Different from dependency graph, each edge in the ER graph is labeled with a pair of relationships from two KBs.

\begin{definition}[ER Graph]
Given two KBs $\kb_1=(\entset_1,\litset_1,\attrset_1,$ $\relset_1,\tripleset_1)$ and $\kb_2=(\entset_2,\litset_2,\attrset_2,\relset_2,\tripleset_2)$, an ER graph on $\kb_1$ and $\kb_2$ is a directed, edge-labeled multigraph $\ergraph=(\vertexset,\edgeset,l_e)$, such that (1) $\vertexset\subseteq \entset_1 \times \entset_2$; (2) for each vertex pair $(\ent_1,\ent_2), (\ent_1',\ent_2')\in V$, $\big((\ent_1,\ent_2), (\ent_1',\ent_2')\big)\in\ergraph\wedge l_v\big((\ent_1,\ent_2), (\ent_1',\ent_2')\big)=(\rel_1, \rel_2)$ if and only if $(\ent_1,\rel_1,\ent_1')\in \tripleset_1\wedge(\ent_2,\rel_2,\ent_2')\in \tripleset_2$.
\end{definition}

Figure~\ref{fig:graph} illustrates an ER graph fragment built from DBpedia and YAGO. Note that, an entity can occur in multiple vertices, and a relationship can appear in different edge labels. A \emph{probabilistic ER graph} is an ER graph where each edge $\big((\ent_1,\ent_2), (\ent_1',\ent_2')\big)$ is labeled with a conditional probability $\Pr(\ent_1'\simeq\ent_2' \,|\, \ent_1\simeq\ent_2)$. The major challenge of constructing an ER graph is how to significantly reduce the size of the graph while preserving as many potential entity matches as possible. %In this section, we propose a $k$-nearest neighbor pruning method based on the partial order on entity similarities.

\subsection{Candidate Entity Match Generation}

We conduct a string matching on entity labels (e.g., the values of \textit{rdfs:label}) to generate \emph{candidate} entity matches and regard them as vertices in the ER graph. Specifically, we first normalize entity labels via lowercasing, tokenization, stemming, etc. Then, we leverage the $\mathtt{Jaccard}$ coefficient---the size of the intersection divided by the size of the union of two sets---as our similarity measure to compute similarities on the normalized label token sets and follow the previous studies \cite{wang2013,chai2018power,wang2012} to prune the entity pairs whose similarities are less than a predefined threshold (e.g., 0.3). Although the choice of thresholds is dataset dependent, this process runs fast and largely reduces the amount of non-matches, thus helping the ER approaches scale up. Note that there are many choices on the similarity metric, e.g., $\mathtt{Jaccard}$, $\mathtt{cosine}$, $\mathtt{dice}$ and edit distance~\cite{sun2017}; our approach can work with any of them and we use $\mathtt{Jaccard}$ for illustration purpose only. The set of candidate entity matches is denoted by $\candidatematches$. 
Similar to \cite{zhuang2017hike}, we use the label similarities as prior match probabilities (i.e., $\mathrm{Pr}[\pmatch]$). More accurate estimation in \cite{whang2013,firmani2016} can be achieved by human labeling.

%--------------------
\subsection{Attribute Matching}
\label{subsect:attribute}

In $\candidatematches$, we refer to the subset of its entities that has exactly the same labels as \emph{initial} entity matches. We leverage them as a priori knowledge for attribute and relationship matching (see Sections~\ref{subsect:attribute} and \ref{subsect:relationship}). Other features, e.g., \textit{owl:sameAs} and inverse functional properties \cite{schneider2004owl}, may also be used to infer initial entity matches~\cite{hu2015,zhuang2016pba}. Note that we do not directly add initial entity matches in the final ER results, because they may contain errors. The set of initial entity matches is denoted by $\initialmatches$.

For such a set of initial entity matches $\initialmatches$ between two KBs $\mathcal{K}_1=(U_1,L_1,A_1,R_1,T_1)$ and $\mathcal{K}_2=(U_2,L_2,A_2,R_2,$ $T_2)$, we proceed to define the following attribute similarity to find their attribute matches. For any two attributes $a_1\in A_1$ and $a_2\in A_2$, their similarity $\mathtt{sim}(a_1, a_2)$ is defined as the average similarity of their values:
\begin{align}
    \mathtt{sim}_A(\attr_1, \attr_2) = \frac{\sum_{(\ent_1, \ent_2)\in \initialmatches} \mathtt{sim}_L(\valueset_1,\valueset_2)}{\big| \{(\ent_1,\ent_2)\in \initialmatches : \valueset_1\cup \valueset_2\neq \emptyset\} \big|},
\end{align}
where $\valueset_1=\{\lit_1: (\ent_1,\attr_1,\lit_1) \in T_1\}$ and $\valueset_2$ is defined analogously. $\mathtt{sim}_L$ represents an extended $\mathtt{Jaccard}$ similarity measure for two sets of literals, which employs an internal literal similarity measure and a threshold to determine two literals being the same when their similarity is not lower than the threshold~\cite{naumann2010}. For different types of literals, we use the $\mathtt{Jaccard}$ coefficient for strings and the maximum percentage difference for numbers (e.g., integers, floats and dates). The threshold is set to 0.9 to guarantee high precision. We refer interested readers to~\cite{cheatham2014} for more information about attribute matching.

For simplicity, every attribute in one KB is restricted to match at most one attribute in the other KB. This global 1:1 matching constraint is widely used in ontology matching~\cite{megdiche2016}, and facilitates our assembling of similarity vectors (later in  Section~\ref{sec:partial}). The 1:1 attribute matching selection is modeled as the bipartite graph matching problem and solved with the Hungarian algorithm~\cite{kuhn1955hungarian} in $O((|A_1| + |A_2|)^2|A_1||A_2|)$ time. The set of attribute matches is denoted by $\attrmatches$.

%--------------------
\subsection{Partial Order Based Pruning}
\label{sec:partial}

Given the candidate entity match set $\candidatematches$ and the attribute match set $\attrmatches$, for each candidate $(\ent_1,\ent_2)\in \candidatematches$, we create a similarity vector $\mathbf{s}(u_1, u_2)=(s_1,s_2,\ldots,s_{|M_{at}|})$, where $s_i$ is the literal similarity ($\mathtt{sim}_L$) between $u_1$ and $u_2$ on the $i^\textrm{th}$ attribute match ($1\le i\le |M_{at}|$). As a consequence, a natural partial order exists among the similarity vectors: $\mathbf{s} \succeq \mathbf{s}'$ if and only if $\forall 1\leq i\leq |M_{at}|, \mathbf{s}_i \ge \mathbf{s}_i'$. This partial order can be used to determine whether an entity pair is a (non-)match in two ways: (i) an entity pair $(\ent_1,\ent_2)$ is a match if there exists an entity pair $(\ent_1',\ent_2')$ such that $(\ent_1',\ent_2')$ is a match and $\mathbf{s}(\ent_1,\ent_2) \succeq \mathbf{s}(\ent_1',\ent_2')$; and (ii) $(\ent_1,\ent_2)$ is a non-match if there exists $(\ent_1',\ent_2')$ such that $(\ent_1',\ent_2')$ is a non-match and $\mathbf{s}(\ent_1',\ent_2') \succeq \mathbf{s}(\ent_1,\ent_2)$.

We incorporate this partial order into a $k$-nearest neighbor search for further pruning the candidate entity match set $M_c$. Let us assume that an entity $u_1$ in one KB has a set of candidate match counterparts $\{u_2^1,u_2^2,\dots,u_2^J\}$ in another KB. The similarity vectors are written as $\mathbf{s}(u_1,u_2^1), \mathbf{s}(u_1,u_2^2),\dots,$ $\mathbf{s}(u_1,u_2^J)$, and we want to determine the top-$k$ in them. Since the partial order is a weak ordering, we count the number of vectors strictly larger than each pair $(u_1, u_2^j)\ (1\le j \le J)$ as its ``rank'', i.e, the minimal rank in all possible  refined full orders. Note that the counterparts of entities in one entity pair are both considered. So, the worst rank of an entity pair $(u_1, u_2)$, denoted by $\mathtt{min\_rank}(u_1, u_2)$, is
\begin{align}
\begin{split}
    \mathtt{min\_rank}(u_1, u_2) &= \max_{i\in\{1,2\}} \mathtt{min\_rank}_i(u_1, u_2) \ , \\
    \mathtt{min\_rank}_1(u_1,u_2) &= \big| \{u_2':\mathbf{s}(u_1,u_2')\succ\mathbf{s}(u_1,u_2)\} \big| \ ,\\
    \mathtt{min\_rank}_2(u_1,u_2) &= \big| \{u_1':\mathbf{s}(u_1',u_2)\succ\mathbf{s}(u_1,u_2)\} \big| \ ,
\end{split}
\end{align}
where all $(u_1, u_2), (u_1,u_2'), (u_1',u_2) \in M_c$.

By $\mathtt{min\_rank}$, we design a modified $k$-nearest neighbor algorithm on this partial order (see Algorithm~\ref{algo:pop}). Because the full order among candidate entity matches is unknown, instead of finding the top-$k$ matches directly, we prune the ones that cannot be in top-$k$. Thus, each entity pair $(u_1, u_2) \in M_c$ such that $\mathtt{min\_rank}(u_1,u_2) \geq k$ needs to be pruned. Also, each pair smaller than a pruned pair should be removed based on the partial order to avoid redundant checking, because $\mathtt{min\_rank}$ of these pairs must be greater than $k$. The set of retained entity matches is denoted by $M_{rd}$, where each entity is involved in \emph{nearly} $k$ candidate matches, due to the weak ordering of partial order.

Algorithm~\ref{algo:pop} first partitions entity match set $M$ into each block $B$ where all pairs contain the same entity (Line 8). Then, it checks each entity pair $(u_1, u_2)\in B$, and prunes entity pairs such that $\mathtt{min\_rank} \geq k$ (Lines 10--12). Finally, the retained pairs in $B$ are added into the output match set.

Algorithm~\ref{algo:pop} first takes $\mathrm{O}(|M_c||M_{at}|)$ time to pre-compute the similarity vectors. When processing $U_i$ $(i=1,2)$, the pruning step (Lines 7--13) checks at most $|M_c|$ pairs, and each time it spends $\mathrm{O}(3|U_{3-i}||M_{at}|)$ time to compute $\mathtt{min\_rank}_i$, prune pairs in $\partition$ and store the retained pairs in $D$. So, the overall time complexity of Algorithm~\ref{algo:pop} is $\mathrm{O}\big(|M_c||M_{at}|(|U_1|+|U_2|)\big)$. In practice, similarity vector construction is the most time-consuming part, while the pruning step only needs to check a small amount of entities in $U_1$ or $U_2$.

\section{Relational Match Propagation}
\label{sect:propagate}

Given an ER graph $\ergraph=(\vertexset,\edgeset,l_v,l_e)$ and an entity match $\ent_1\simeq\ent_2$ in it, the relational match propagation infers how likely each unresolved entity pair $\pair\in\vertexset$ is a match based on the structure of $\ergraph$, i.e. $\mathrm{Pr}[\pmatch \,|\, \ent_1\simeq\ent_2]$. In this section, we first consider a basic case that unresolved entity pairs are neighbors of a match in $\ergraph$. Then, we generalize it to the case that unresolved pairs are reachable from several matches. In the basic case, we resolve entity pairs between two value sets of a relationship pair, and define the consistency between relationships to measure the portion of values containing matched counterparts in another value set. The consistency and the prior match probabilities of entity pairs are further combined to obtain ``tight'' posterior match probabilities. In the general case, we propose a Markov model on paths from matches to unresolved ones to find the match probability bounds.

\begin{algorithm}[t]
   {\footnotesize
      \KwIn{Candidate entity match set $M_c$, attribute match set $M_{at}$, threshold $k$}
      \KwOut{Retained entity match set $M_{rd}$}
      \SetKwProg{Fn}{Function}{}{}
      \SetKwFunction{PIOW}{PruningInOneWay}
   
      \lForEach{$(u_1,u_2)\in M_c$}{pre-compute $\mathbf{s}(u_1,u_2)$}
          $M_{rd} \leftarrow$ \PIOW{$M_c, U_1, k$}\;
          $M_{rd} \leftarrow$ \PIOW{$M_{rd}, U_2, k$}\;
          \Return $M_{rd}$\;
      \BlankLine
       
          \Fn{\PIOW{$M, U_i, k$}}{
         $D\leftarrow\emptyset$\;
              \ForEach{$u_i\in U_i$}{
                     $B \leftarrow \big\{(u_1, u_2)\in M : u_1=u_i\vee u_2=u_i\big\}$\;
            \lIf(\tcc*[f]{no need to prune}){$|B| \leq k$}{continue}
            %$B' \leftarrow B$\;
                     \ForEach{$(u_1, u_2)\in B$}{
                            \If{$\mathtt{min\_rank}_i(u_1,u_2)\geq k$} {
                  \tcc{$(u_1',u_2')$ cannot be pruned here}
                                $B \leftarrow \big\{(u_1',u_2')\in B : \mathbf{s}(u_1,u_2) \not\succeq \mathbf{s}(u_1',u_2')\big\}$\;
                            }
                     }
                     $D \leftarrow D \cup B$\;
              }
              \Return $D$;
          }
   }
   \caption{Partial order based pruning}
   \label{algo:pop}
   \end{algorithm}
   
%--------------------
\subsection{Consistency Between Relationships}
\label{subsect:relationship}

Functional/inverse functional properties are ideal for match propagation. For example, \emph{wasBornIn} is a functional property, and the born places of two persons in a match must be identical. However, we cannot just rely on functional/inverse functional properties, since many relationships are multi-valued and only a part of the values may match. Thus, we define the consistency between relationships as follows.

Let $\rel_1$ and $\rel_2$ be two relationships in two KBs . We assume that, given the condition that $\ent_1\simeq \ent_2 \wedge \ent_1'\in \neighborset_1$, the probability of the event $\exists \ent_2'~:~\big(\ent_2'\in \neighborset_2 \wedge \ent_1'\simeq \ent_2'\big)$ is subject to a binary distribution with parameter $\consistency_1$. Symmetrically, we define parameter $\consistency_2$. We use $\consistency_1$ and $\consistency_2$ to represent the consistency for two relationships $\rel_1$ and $\rel_2$, respectively:
\begin{align}
\resizebox{.9\columnwidth}{!}{$\begin{aligned}
\consistency_1 &= \mathrm{Pr}[\exists \ent_2' : \ent_2'\in \neighborset_2 \wedge \ent_1'\simeq \ent_2' \,|\, \ent_1\simeq \ent_2,\ent_1'\in \neighborset_1] , \\
\consistency_2 &= \mathrm{Pr}[\exists \ent_1' : \ent_1\in \neighborset_1 \wedge \ent_2'\simeq \ent_1' \,|\, \ent_2\simeq \ent_1,\ent_2'\in \neighborset_2] .
\end{aligned}$}
\end{align}
where $\neighborset_1,\neighborset_2$ are the value sets of relationships $\rel_1,\rel_2$ w.r.t. entities $\ent_1,\ent_2$, respectively. 

To estimate $\consistency_1$ and $\consistency_2$, we use the value distribution on the initial entity matches $\initialmatches$. For an entity pair $(\ent_1, \ent_2)\in \initialmatches$, we introduce a latent random variable $\neighbormatchcnt^{r_1,r_2}= |\neighbormatch^{r_1,r_2}|$, where $\neighbormatch^{r_1,r_2}$ denotes the set of entity matches in $\neighborset_1\times \neighborset_2$. Note that we omit $r_1, r_2$ in $\neighbormatchcnt^{r_1,r_2}$ and $\neighbormatch^{r_1,r_2}$ to simplify notations. Similar to~\cite{zhang2015}, we make an assumption on the entity sets: no duplicate entities exist in each entity set. Hence, $\neighbormatchcnt$ is also the number of entities in $\neighborset_1$ (or $\neighborset_2$) which appear in $\neighbormatch$. Based on the latent variable $\neighbormatchcnt$, the likelihood probability of $(\neighborset_1, \neighborset_2, \neighbormatchcnt)$ is
\begin{equation}
   \resizebox{.9\columnwidth}{!}{$\mathrm{Pr}[\neighborset_1,\neighborset_2,\neighbormatchcnt] 
   = \prod\limits_{i=1,2}\binom{| \neighborset_i |}{\neighbormatchcnt}(\frac{ \consistency_i}{1 - \consistency_i})^{\neighbormatchcnt}(1-\consistency_i)^{|\neighborset_i|}$.}
\end{equation}

Then, we use the maximum likelihood estimation to obtain $\consistency_1$ and $\consistency_2$:
\begin{align}
\label{eq:mle}
\max_{\consistency_1, \consistency_2, L_{\cdot,\cdot}} \prod_{(\ent_1, \ent_2) \in \initialmatches} \mathrm{Pr}[\neighborset_1,\neighborset_2,\neighbormatchcnt] .
\end{align}

Since each $\neighbormatchcnt$ is an integer variable, the brute-force optimization can cost exponential time. Next, we present an optimization process. Let $\zeta = \frac{\consistency_1 \consistency_2}{(1 - \consistency_1)(1 - \consistency_2)}$ and $\xi(\consistency_1, \consistency_2)=(1-\consistency_1)^{b_1} (1-\consistency_2)^{b_2}$, where $b_1=\sum|\neighborset_1|, b_2=\sum|\neighborset_2|$. We simplify (\ref{eq:mle}) to $\max_{\consistency_1, \consistency_2} \xi(\consistency_1, \consistency_2)\prod\max_{\neighbormatchcnt}c_{\neighbormatchcnt}\zeta^{\neighbormatchcnt}$, where $c_{\neighbormatchcnt} = \binom{| \neighborset_1 |}{\neighbormatchcnt}\binom{| \neighborset_2 |}{\neighbormatchcnt}$. Notice that $c_i \zeta^i = c_j \zeta^j$ has only one solution for different integers $i, j$. Thus, the curves $c_{\neighbormatchcnt}\zeta^{\neighbormatchcnt}$ ($0\leq \neighbormatchcnt\leq L_M$) can have at most $\binom{L_M + 1}{2}$ common points, where $L_M=\min\{|\neighborset_1|, |\neighborset_2|\}$. Therefore, $\max_{\neighbormatchcnt}c_{\neighbormatchcnt}\zeta^{\neighbormatchcnt}$ is an $\mathrm{O}(L_M^2)$-piecewise continuous function, and the product of these $\mathrm{O}(L_M^2)$-piecewise continuous functions is an $\mathrm{O}(\max\{|\neighborset_1|^4,|\neighborset_2|^4\})$-piecewise continuous function. As a result, we can optimize (\ref{eq:mle}) by solving $\mathrm{O}(\max\{|\neighborset_1|^4,|\neighborset_2|^4\})$ continuous optimization problems with two variables, which runs efficiently.

%--------------------
\subsection{Match Propagation to Neighbors}

A basic case is that the unresolved entity pairs are adjacent to a match $\ent_1\simeq\ent_2$ in $\ergraph$. We consider the neighbors with the same edge label, i.e. relationship pair $(\rel_1, \rel_2)$, together. Then, our goal is to identify matches between $\neighborset_1$ and $\neighborset_2$. 

Let $\neighbormatch\subseteq \neighborset_1\times \neighborset_2$ denote a set of entity matches. We consider two factors about how likely $\neighbormatch$ can be the correct match result of $\neighborset_1\times \neighborset_2$: (1) the prior match probabilities of matches without neighborhood information; (2) the consistency of the relationships. The match probability of $\neighbormatch$ given $\ent_1\simeq \ent_2$ is: 
\begin{align}
\mathrm{Pr}[\neighbormatch \,|\, \ent_1\simeq \ent_2] &= \frac{1}{Z}\, f(\neighbormatch \,|\, \neighborset_1,\neighborset_2) \notag\\ 
\times &\, g(\neighbormatch \,|\, \neighborset_1)\, g(\neighbormatch \,|\, \neighborset_2) ,
\end{align}
where $Z$ is the normalization factor. $f(\neighbormatch \,|\, \neighborset_1,\neighborset_2)$ is the prior match probability. $g(\neighbormatch \,|\, \neighborset_1), g(\neighbormatch \,|\, \neighborset_2)$ are the consistency of $\neighbormatch$ w.r.t.~$\neighborset_1,\neighborset_2$, respectively.

Without considering neighborhood information, the prior match probability $f(\neighbormatch \,|\, \neighborset_1,\neighborset_2)$ is defined as the likelihood function of $\neighbormatch$:
\begin{align}
\resizebox{.9\columnwidth}{!}{$\begin{aligned}
f(\neighbormatch \,|\, \neighborset_1,\neighborset_2) = \smashoperator[r]{\prod_{p\in \neighbormatch}}\mathrm{Pr}[\pmatch]\times \smashoperator[r]{\prod_{p\in \neighborset_1\times \neighborset_2\setminus \neighbormatch}}(1-\mathrm{Pr}[\pmatch]) ,
\end{aligned}$}
\end{align}
where $\mathrm{Pr}[\pmatch]$ denotes the prior probability of entity pair $\pair$ being a match, and $1-\mathrm{Pr}[\pmatch]$ denotes the prior probability of $\pair$ being a non-match.

Let $\pi_1(\neighbormatch)=\{\ent_1' \,|\, (\ent_1', \ent_2')\in \neighbormatch\}$. Note that when $\ent_1$ and $\ent_2$ form a match, each entity $\ent_1'\in\pi_1(\neighbormatch)$ is a neighbor of $\ent_1$ for relationship $\rel_1$ such that $\exists \ent_2':\ent_2'\in\neighborset_2\wedge\ent_2'\simeq \ent_1'$. Based on $\consistency_1$, the consistency of $\neighbormatch$ given $\neighborset_1$ is defined as follows:
\begin{align}
\resizebox{.9\columnwidth}{!}{$\begin{aligned}
g(\neighbormatch \,|\, \neighborset_1) = \consistency_1^{|\pi_1(\neighbormatch)|}\, (1-\consistency_1)^{|\neighborset_1| - |\pi_1(\neighbormatch)|} .
\end{aligned}$}
\end{align}
$\pi_2(\neighbormatch)$ and $g(\neighbormatch \,|\, \neighborset_2)$ can be defined similarly.

Finally, we obtain the posterior match probability of $\ent_1'\simeq \ent_2'$ by marginalizing $\mathrm{Pr}[\ent_1'\simeq \ent_2', \neighbormatch \,|\, \ent_1\simeq \ent_2]$:
\begin{align}
\label{eq:mp_neighbor}
\resizebox{\columnwidth}{!}{$\begin{aligned}
&\mathrm{Pr}[\ent_1'\simeq \ent_2' \,|\, \ent_1\simeq \ent_2]  \\
&= \smashoperator{\sum\limits_{\neighbormatch}} \mathrm{Pr}[\ent_1'\simeq \ent_2', \neighbormatch \,|\, \ent_1\simeq \ent_2] = \smashoperator{\sum_{\neighbormatch :\, (\ent_1', \ent_2')\in \neighbormatch}}\mathrm{Pr}[\neighbormatch \,|\, \ent_1\simeq \ent_2], %\\
%&= \frac{\sum_{\neighbormatch :\, (\ent_1', \ent_2')\in \neighbormatch} \mathrm{Pr}[\neighbormatch \,|\, \ent_1\simeq \ent_2]}{\sum_{\neighbormatch}\mathrm{Pr}[\neighbormatch \,|\, \ent_1\simeq \ent_2]} ,
\end{aligned}$}
\end{align}
where $\neighbormatch$ is selected over $(\neighborset_1\times \neighborset_2)\cap \vertexset$.%, and the last ``='' holds because $Z$ makes $\sum_{\neighbormatch}\mathrm{Pr}[\neighbormatch \,|\, \ent_1\simeq \ent_2]$ equal to $1$. We drop $Z$ when evaluating the last fraction.

\noindent\textbf{Example.} Let $(\ent_1, \ent_2)=(\textit{y:Tim}, \textit{d:Tim})$, $\rel_1$ and $\rel_2$ denote the relationship \textit{directed}, $\consistency_1 = \consistency_2 = 0.9$, and $\mathrm{Pr}[\pmatch] \equiv 0.5$ (implying all pairs are viewed as the same). From Figure~\ref{fig:graph}, we can find that $\neighborset_1=\{\textit{y:Cradle}, \textit{y:Player}\}$ and $\neighborset_2=\{\textit{d:Cradle}, \textit{d:Player}\}$. Thus, when $\neighbormatch = \{(\textit{y:Cradle}, \textit{d:Cradle}), (\textit{y:Player}, \textit{d:Player})\}$, $\mathrm{Pr}[\neighbormatch\,|\,\ent_1\simeq\ent_2]=0.5^3\times 0.95^4\approx 0.1$; when $\neighbormatch' = \{(\textit{y:Cradle},\textit{d:Player})\}$, $\mathrm{Pr}[\neighbormatch'\,|\,\ent_1\simeq\ent_2]=0.5^3\times 0.95^2\times 0.05^2\approx 0.0003$. So, $\neighbormatch$ is more likely to be the match set within $\neighborset_1\times\neighborset_2$. Furthermore, $\mathrm{Pr}[\textit{y:Cradle}\simeq\textit{d:Cradle}]\approx 0.99$, whereas $\mathrm{Pr}[\textit{y:Cradle}\simeq\textit{d:Player}]\approx 0.01$.
%--------------------
\subsection{Distant Match Propagation} 

The above match propagation to neighbors only estimates the match probabilities of direct neighbors of an entity match, which lacks the capability of discovering entity matches far away. In the following, we extend it to a more general case, called \emph{distant match propagation}, where a match reaches an unresolved entity pair through a path.

Intuitively, given a match $(\ent_1,\ent_2)$ and an unresolved pair $(\ent_1',\ent_2')$, the distant propagation process can be modeled as a path consisting of the entity pairs from $(\ent_1,\ent_2)$ to $(\ent_1',\ent_2')$, where each unresolved pair can be inferred as a match via its precedent. Assume that there is a path $(\ent_1^0,\ent_2^0), (\ent_1^1,\ent_2^1),$ $\dots,(\ent_1^l,\ent_2^l)$ in $\ergraph$, where $(\ent_1^0,\ent_2^0)=(\ent_1,$ $\ent_2)$ and $(\ent_1^l,\ent_2^l)=(\ent_1',\ent_2')$. According to the chain rule of conditional probability, we have
\begin{align}
\label{eq:chain}
\begin{split}
&\mathrm{Pr}[\ent_1^l\simeq \ent_2^l \,|\, \ent_1^0\simeq \ent_2^0] \\
&\geq \mathrm{Pr}[\ent_1^l\simeq \ent_2^l,\ent_1^2\simeq \ent_2^2,\dots,\ent_1^l\simeq \ent_2^l \,|\, \ent_1^0\simeq \ent_2^0] \\
&= \prod\nolimits_{i=1}^l\mathrm{Pr}[\ent_1^i\simeq \ent_2^i \,|\, \ent_1^0\simeq \ent_2^0,\dots,\ent_1^{i-1}\simeq \ent_2^{i-1}] \\
&= \prod\nolimits_{i=1}^{l}\mathrm{Pr}[\ent_1^i \simeq \ent_2^i \,|\, \ent_1^{i-1} \simeq \ent_2^{i-1}],
\end{split}
\end{align}
where the last ``='' holds because we assume that this propagation path satisfies the Markov property~\cite{rastogiL2011mln}. Inequation (\ref{eq:chain}) gives a lower bound for $\mathrm{Pr}[\ent_1^l\simeq \ent_2^l \,|\, \ent_1^0\simeq \ent_2^0]$.
The largest lower bound is selected to estimate $\mathrm{Pr}[\ent'_1\simeq \ent'_2 \,|\, \ent_1\simeq \ent_2]$. We estimate $\mathrm{Pr}[\ent'_1\simeq \ent'_2 \,|\, \ent_1\simeq \ent_2]$ in Algorithm~\ref{algo:imsd}.

%To achieve this, we search all possible paths from $(\ent_1, \ent_2)$ to $(v_1, v_2)$ and compare the lower bounds. 

% Since our relational match propagation only estimates the probabilities, we define the entity pairs that can be resolved by entity pair $\pair$ as
% \begin{align}
% \inferred(\pair) = \{ \pair'\in \vertexset : \mathrm{Pr}[\match_{\pair'} \,|\, \pmatch]\ge\tau \} ,
% \end{align}
% where $\tau$ is the precision threshold for inferring high-quality matches. Note that this formulation is different from estimating the precision based on the average match probability among all entity pairs that are resolved as matches~\cite{whang2013,qian2017erlearn}.

\section{Multiple Questions Selection}
\label{sect:question}

Based on the relational match propagation, unresolved entity pairs can be inferred by human-labeled matches. However, different questions have different inference capabilities. In this section, we first describe the definition of inferred match set and the multiple questions selection problem. Then, we design a graph-based algorithm to determine the inferred match set for each question. Finally, we formulate the benefit of multiple questions and design a greedy algorithm to select the best questions.

\subsection{Question Benefits}

We follow the so-called \emph{pairwise question interface}~\cite{wang2013,whang2013,firmani2016,zhuang2017hike,vesdapunt2014,verroios2015}, where each question is whether an entity pair is a match or not. 
%This type of question is regarded easier for workers compared with others like clustering multiple entities~\cite{wang2012,verroios2017waldo}. 
Let $\questionset$ be a set of pairwise questions. Labeling $\questionset$ can be defined as a binary function $\answerset : \questionset \rightarrow \{0,1\}$, where for each question $\question \in \questionset$, $\answerset(\questionset)=1$ means that $\question$ is labeled as a match, while $\answerset(\question)=0$ indicates that $\question$ is labeled as a non-match. 

Given the labels $\answerset$, we propagate the labeled matches in $H$ to unresolved pairs. The set of entity pairs that can be inferred as matches by $\answerset$ is
\begin{align}
\label{eq:benefit_H}
	\inferred(\answerset) & = \bigcup\nolimits_{\question\in\questionset : \answerset(\question)=1}\inferred(\question) , \\
	\inferred(\question) & = \{ \pair\in \unresolved : \mathrm{Pr}[\pmatch \,|\, \qmatch]\ge\tau \},\label{eq:inf_q}
\end{align}
where $\unresolved$ is the unresolved entity pairs and $\tau$ is the precision threshold for inferring high-quality matches. We evaluate $\inferred(\question)$ in Section~\ref{sec:discovery}.

%Since our relational match propagation only estimates the probabilities, we define the entity pairs that can be resolved by entity pair $\pair$ as
%\begin{align}
%\inferred(\pair) = \{ \pair'\in \vertexset : \mathrm{Pr}[\match_{\pair'} \,|\, \pmatch]\ge\tau \} ,
%\end{align}
%where $\tau$ is the precision threshold for inferring high-quality matches. Note that this formulation is different from estimating the precision based on the average match probability among all entity pairs that are resolved as matches~\cite{whang2013,qian2017erlearn}.

Since non-matches are quadratically more than matches in the ER problem~\cite{getoor2012}, the labels to the ideal questions should infer as many matches as possible. Thus, we define the benefit function of $\questionset$ as the expected number of matches can be inferred by labels to $\questionset$, which is
\begin{align}
	\benefit(\questionset) =\mathtt{E}\big[|\inferred(\answerset)| \,\big|\, \questionset\big] .
\end{align}

The ER algorithm can ask each question with the greatest $\benefit$ iteratively; however, there is a latency caused by waiting for workers to finish the question. Assigning multiple questions to workers simultaneously in one human-machine loop is a straightforward optimization to reduce the latency. Since workers in crowdsourcing platforms are paid based on the number of solved questions, the number of questions should be smaller than a given budget. Thus, the \emph{optimal multiple questions selection} problem is to
\begin{align}
\begin{split}
\mathrm{maximize}\quad & \benefit(\questionset) , \\
\mathrm{s.t.\ \ \ \ }\quad & \questionset \subseteq \unresolved,\enskip  |\questionset| \leq \qlimit ,
\end{split}
\end{align}
where $\qlimit$ is the constraint on the number of questions asked.

%--------------------
\subsection{Discovery of Inferred Match Set}
\label{sec:discovery}
In order to obtain the $\benefit$ for each question set $\questionset$, we need to compute $\inferred(\question)$ for each $q \in Q$. To estimate  $\mathrm{Pr}[\pmatch\,|\,\qmatch]$ in $\inferred(\question)$, we define the length of a directed edge $(\vertex,\vertex')$ in probabilistic ER graph $\pergraph$ as $\mathtt{length}(\vertex,\vertex') = -\log f(\vertex,\vertex')=-\log\mathrm{Pr}[\vvmatch\,|\,\vmatch]$. According to the definition of $\mathrm{Pr}[\pmatch\,|\,\qmatch]$, $\mathrm{Pr}[\pmatch\,|\,\qmatch]=e^{\mathtt{dist}(q,p)}$, where $\mathtt{dist}(q,p)$ is the distance of the shortest path from $q$ to $p$. As a result, the condition $\mathrm{Pr}[\pmatch \,|\, \qmatch]\geq \tau$ can be interpreted as $\mathtt{dist}(q,p)\leq \zeta=-\log\tau$. Note that edge $(\vertex,\vertex')$ can be removed when $\mathrm{Pr}[\vvmatch\,|\,\vmatch] = 0$ to avoid $\log 0$.

The all-pairs shortest path algorithms can efficiently compute $\inferred(\question)$ for every $\question$. Since most $|\inferred(\question)|$ should be smaller than $|\unresolved|$, we choose to apply binary trees rather than an array of size $|\unresolved|$ to maintain distances. We depict our modified Floyd-Warshall algorithm in Algorithm~\ref{algo:imsd}. In Lines 1--2, for every $\question$, we create a binary tree $\mathtt{bt}(\question)$ to store the inferred pairs as well as their corresponding lengths, and a binary tree $\mathtt{bt}^{-1}(\question)$ to store pairs inferring $\question$ as well as their corresponding lengths. In Lines 3--5, the edge whose length is not greater than $\zeta$ would be stored into binary trees. In Lines 6--11, we modify the dynamic programming process in the original Floyd–Warshall algorithm. Since the number of pairs which can be inferred is significantly less than $|\unresolved|$, the inner loop in Lines 9--11 iterate only over the set of distances which are likely to be updated. Lines 13--14 extract the inferred match sets from binary trees. 

Since each binary tree contains at most $|C|$ elements, $|R| \leq |\mathtt{bt}(\question).val| \leq |\unresolved|$. The loop in Lines 6--11 takes $\mathrm{O}(|\unresolved|^3)$ time in total. The time complexity of Algorithm~\ref{algo:imsd} is $\mathrm{O}(|\unresolved|^3)$.
 
\begin{algorithm}[t]
{\footnotesize
\KwIn{Probabilistic ER graph $\pergraph$, candidate question set $\unresolved$, distance threshold $\zeta$}
\KwOut{Set $B$ of inferred match sets for all questions}
\ForEach{$q\in \unresolved$}{
	Initialize two empty binary trees $\mathtt{bt}(q),\mathtt{bt}^{-1}(q)$\;
}
\ForEach{$(q,p)\in \unresolved\times \unresolved$}{
	\If{$\mathtt{length}(q,p)\leq\zeta$}{
		$\mathtt{bt}(q)[p] \leftarrow \mathtt{length}(q,p); \mathtt{bt}^{-1}(p)[q] \leftarrow \mathtt{length}(p,q)$\;
    	}
}
\ForEach{$q\in \unresolved$}{
	\ForEach{$\pair\in\mathtt{bt}(q).val$}{
		$R\leftarrow\{r\in\mathtt{bt}^{-1}(q).val : \mathtt{bt}(q)[p]+\mathtt{bt}^{-1}(q)[r]\leq\zeta\}$\;
		\ForEach{$r\in R$}{
        	       		$d\leftarrow\mathtt{bt}(q)[p] + \mathtt{bt}^{-1}(q)[r]$\;
        	        		$\mathtt{bt}(q)[r] \leftarrow d; \mathtt{bt}^{-1}(r)[p] \leftarrow d$\;
    	        }
	}
}
$B\leftarrow\emptyset$\;
\ForEach{$q\in \unresolved$}{
	$\inferred(\question)\leftarrow\mathtt{bt}(q).val; B\leftarrow B\cup\{\inferred(\question)\}$\;
}
\Return $B$\;
}
\caption{DP-based inferred match set discovery}
\label{algo:imsd}
\end{algorithm}

%--------------------
\subsection{Multiple Questions Selection}

Since the match propagation works independently for each label, the event that an entity pair $\pair$ is inferred as a match by labels $\answerset$ is equivalent to the event that $\pair$ is inferred by $\question\in\questionset$ such that $\answerset(\question) = 1$. When $\answerset$ is not labeled, $\pair$ is resolved as a match if and only if at least one question that can resolve $\pair$ as a match is labeled as a match. Given the question set $\questionset$, the probability that $\pair$ can be resolved as a match by labels is
\begin{align}
\mathrm{Pr}[\pair\in\inferred(\answerset) \,|\, \questionset] = 1-\quad\smashoperator{\prod_{\question\in\questionset : \pair\in\inferred(\question)}}(1-\mathrm{Pr}[\qmatch]) ,
\end{align}
where $\inferred(\answerset)$ is defined in Eq.~(\ref{eq:benefit_H}), representing the matches that can be inferred after $\questionset$ is labeled by workers. 

The benefit of question set $\questionset$ is formulated as the expected size of the inferred matches by labels $\answerset$:
\begin{align}
\label{eq:benefit_Q}
\begin{split}
\benefit(\questionset) & = \mathtt{E}\big[|\inferred(\answerset)| \,\big|\, \questionset\big] \\
 & = \sum_{\pair\in \unresolved}\mathrm{Pr}\Big[\pair\in \inferred(\answerset) \,|\, \questionset\Big].
\end{split}
\end{align}

Now, we want to select a set of questions that can maximize the benefit. We first prove the hardness of the multiple questions selection problem. Then, we describe a greedy algorithm to solve it.
\begin{theorem}
\label{the:np_hard}
The problem of optimal multiple questions selection is NP-hard.
\end{theorem}
\begin{proof}
The optimization version of the set cover problem is NP-hard. Given an element set $U=\{1,2,\dots,n\}$ and a collection $S$ of sets whose union equals $U$, the set cover problem aims to find the minimum number of sets in $S$ whose union also equals $U$. This problem can be reduced to our multiple questions selection problem in polynomial time. Assume that the vertex set of an ER graph is $\{p_1,p_2,\dots,$ $p_n\} \cup \{p_s \,|\, s\in S\}$, the edge set is $\{(p_s,p_k) : k=1,2,\dots,n \wedge s\in S \wedge k\in s\}$, all the prior match probabilities are $1$, the precision threshold is $1$, $\mathrm{Pr}[p_k \,|\, p_s]=1$ and $\mathrm{Pr}[p_s \,|\, p_k]=0$, for all $k=1,2,\dots, n, s\in S$ satisfying $k\in s$. Because the benefit is equal to the number of covered elements in $U$, the optimal solution of the multiple questions selection problem is also that of the set cover problem. Thus, the multiple questions selection problem is NP-hard.
% See Appendix~\ref{app:np}.
\end{proof}

\begin{theorem}
\label{the:submod}
$\benefit(Q)$ is an increasing submodular function.
\end{theorem}
\begin{proof}
% See Appendix~\ref{app:submod}.
Let $\mathtt{b}_p(Q)$ represent $\mathrm{Pr}[p \in \inferred(H) \,|\, Q]$. For every $p \in C$ and two disjoint subsets $Q, Q' \subseteq C$, we have 
\begin{align}
    \mathtt{b}_p(Q \cup Q') = \mathtt{b}_p(Q) + \mathtt{b}_p(Q') - \mathtt{b}_p(Q)\mathtt{b}_p(Q') \ . \notag
\end{align}

Thus, $\mathtt{b}_p(Q \cup Q') - \mathtt{b}_p(Q) = \mathtt{b}_p(Q') \big(1 - \mathtt{b}_p(Q')\big) \geq 0$. Since $\mathtt{benefit}(Q) = \sum_{p \in C} \mathtt{b}_p(Q)$, it is an increasing function.

Also, for every $p \in C, Q \subseteq C$ and $q_1, q_2 \in C \setminus Q$ such that $q_1 \neq q_2$, we have
\begin{align}
\resizebox{\columnwidth}{!}{$\begin{aligned}
& \mathtt{b}_p(Q \cup \{q_1,q_2\}) + \mathtt{b}_p(Q) - \mathtt{b}_p(Q \cup \{q_1\}) - \mathtt{b}_p(Q \cup \{q_2\}) \\
& = b_p(\{q_2\}) \big(\mathtt{b}_p(Q) - \mathtt{b}_p(Q \cup \{q_1\})\big) \leq 0 .
\end{aligned}$} \notag
\end{align}

Thus, $\mathtt{b}_p(Q \cup \{q_1\}) + \mathtt{b}_p(Q \cup \{q_2\}) \geq \mathtt{b}_p(Q \cup \{q_1,q_2\}) + \mathtt{b}_p(Q)$. Since $\mathtt{benefit}(Q) = \sum_{p \in C} \mathtt{b}_p(Q)$, it is a submodular function. Together, we prove that $\mathtt{benefit}(\cdot)$ is an increasing submodular function.
\end{proof}

Since Eq.~(\ref{eq:benefit_Q}) is monotonic and submodular, the multiple questions selection problem can be solved by using submodular optimization. We design Algorithm~\ref{algo:mqs}, which gives a $(1-\frac{1}{e})$-approximation guarantee. This algorithm selects questions greedily with the highest gain in benefits (i.e. $\qdelta$). We also leverage the lazy evaluation of the submodular function to improve the efficiency~\cite{lazy}. Specifically, we maintain a priority queue $PQ$ over each candidate question $\question$ ordered by the gain in benefits $\qdelta$ in descending order. Based on the submodular property, when the gain in benefits $\qdelta$ of the picked question $\question$ is greater than that of the top question $\question'$ in $PQ$, $\question$ is the question with the largest gain in benefits. We use an array to store $b_p(\questionset)$, such that $\qdelta$ can be obtained in $\mathrm{O}(|\unresolved|)$ time. The overall time complexity of Algorithm~\ref{algo:mqs} is $\mathrm{O}(\mu |\unresolved|^2)$, where $\mu$ is the number of questions asked in each loop and $\unresolved$ is the set of unresolved entity pairs in the ER graph.

\begin{algorithm}[t]
{\footnotesize
\KwIn{Probabilistic ER graph $\pergraph$, candidate question set $\unresolved$, precision threshold $\tau$, question number $\mu$}
\KwOut{Selected question set $\questionset$}
$\questionset\leftarrow\emptyset;$ $PQ\leftarrow\{\big(\question, \benefit(\{\question\})\big)\,|\,\question\in\questionset\}$\;
\While{$|\questionset|<\mu$}{
	$\question,\qdelta\leftarrow PQ.pop(); \question',\qdelta'\leftarrow PQ.top()$\;
	\While{$\qdelta > 0$}{
		$\qdelta\leftarrow\benefit(\questionset\cup\{\question\})-\benefit(\questionset)$\;
		\leIf{$\qdelta \geq \qdelta'$}{
			$\questionset\leftarrow \questionset\cup\{\question\}; $ \textbf{break}\;
		}{$PQ.push\big((\question,\qdelta)\big)$}
		$\question,\qdelta\leftarrow PQ.pop(); \question',\qdelta'\leftarrow PQ.top()$\;
	}
	\lIf{$\qdelta\leq 0$}{\textbf{break}}
}
\Return $\questionset$\;
}
\caption{Greedy multiple questions selection}
\label{algo:mqs}
\end{algorithm}

\section{Truth Inference}
\label{sect:inference}

After the questions are labeled by workers, we design an error-tolerant model to infer truths (i.e. matches and non-matches) from the imperfect labeling, which facilitates updating the (probabilistic) ER graph and resolving isolated entities.

\subsection{Error-Tolerant Inference}

As the labels completed by the workers on crowdsourcing platforms may contain errors, we assign one question to multiple workers and use their labels to infer the posterior match probabilities. We leverage the worker probability model~\cite{zheng2017}, which uses a single real number to denote a worker $\worker$'s quality $\lambda^\worker \in (0, 1]$, i.e. the probability that $\worker$ can correctly label a question. Since crowdsourcing platforms, e.g., Amazon MTurk\footnote{\url{https://www.mturk.com/}} offers a qualification test for their workers, we reuse a worker's precision in this test as her quality. The posterior probability of question $\question$ being a match is
\begin{align}
\label{eq:post}
& \mathrm{Pr}[\qmatch \,|\, \trueworkers, \falseworkers] \notag\\
&= \frac{\mathrm{Pr}[\qmatch]\,\mathrm{Pr}[\trueworkers, \falseworkers \,|\, \qmatch]}{\mathrm{Pr}[\qmatch] \, \mathrm{Pr}[\trueworkers, \falseworkers \,|\, \qmatch] + \mathrm{Pr}[\overline{\qmatch}] \, \mathrm{Pr}[\trueworkers, \falseworkers \,|\, \overline{\qmatch}]} \notag\\
&= \frac{\mathrm{Pr}[\qmatch]}{\mathrm{Pr}[\qmatch] + \mathrm{Pr}[\overline{\qmatch}] \prod\limits_{\worker\in\trueworkers} \frac{1 - \lambda^\worker}{\lambda^\worker} \prod\limits_{\worker\in\falseworkers} \frac{\lambda^\worker}{1 - \lambda^\worker}} ,
\end{align}
where $\trueworkers$ denotes the set of workers labeling $\question$ as a match, and $\falseworkers$ denotes the set of workers labeling $\question$ as a non-match.

We assign two thresholds to filter matches and non-matches based on consistent labels. Entity pairs with a high posterior probability (e.g., $\ge 0.8$) are regarded as matches, while pairs with a low posterior probability (e.g., $\le 0.2$) are non-matches. Others are considered as inconsistent and remain unresolved. One possible reason for the inconsistency is that these questions are too hard. For a hard question $q$, we set $\mathrm{Pr}[\qmatch]$ to $\mathrm{Pr}[\qmatch \,|\, \trueworkers, \falseworkers]$ for reducing its benefit, thereby it is less possible to be asked more times. Next, we infer matches based on the consistent labels and re-estimate the probability of each edge in $\pergraph$ using new matches and non-matches.

%--------------------
\subsection{Inference for Isolated Entity Pairs}

As an exception, there may exist a small amount of isolated entity pairs which do not occur in any relationship triples. In this case, the match propagation cannot infer their truths, and the question selection algorithm has to ask these pairs one by one. To avoid such an inefficient polling, we reuse the similarity vectors and the partial order relations obtained in Section~\ref{sect:graph} to train a classifier for these isolated pairs.

Given an isolated entity pair $\pair$, let $\attrset_\pair$ denote the set of its attribute matches. We define the set $\pairpartition$ of retained matches with similar attributes to $\pair$ by $\pairpartition = \{ \pair' \in M_{rd} : \mathtt{Jaccard}(\attrset_\pair, \attrset_{\pair'})$ $\ge \psi \}$, where $\mathtt{Jaccard}$ calculates the similarity between two sets of attribute matches. $\psi$ is a threshold, and we set $\psi=0.9$ for high precision. Since we only allow matches to propagate in the ER graph, most obtained labels are matches. Therefore, we treat all unresolved pairs in $\pairpartition$ as non-matches to balance the proportions of different labels. 

Next, we use $\pairpartition$ and the labels as training data, and scikit-learn\footnote{\url{https://scikit-learn.org}} to train a random forest classifier with default parameter to predict whether $\pair$ is a match. The random forest finds the unresolved pairs in $\pairpartition$ whose similarity vectors are close to known matches.

\section{Experiments and Results}\label{sect:exp} 
In this section, we conduct a thorough evaluation on the effectiveness of our approach Remp, by comparing with state-of-the-art methods followed by an in-depth investigation on each part of Remp (as outlined in Section~\ref{subsect:overview}). 

\noindent\textbf{Datasets.} We use one benchmark dataset and three real-world datasets widely used in previous work~\cite{zhuang2017hike,chai2018power,suchanek2011paris,julien2013sigma}. Table~\ref{tab:stats} lists their statistics.

\begin{itemize}
\item IIMB is a small, synthetic benchmark dataset in OAEI\footnote{\url{http://islab.di.unimi.it/content/im_oaei/2019/}} containing two KBs with identical attributes and relationships. 

\item DBLP-ACM (abbr. D-A)\footnote{\url{https://dbs.uni-leipzig.de/en/research/projects/object_matching}} is a dataset about publications and authors. The original version uses a text field to store all authors of a publication. Here, we split it and create authorship triples. In the case that an author has multiple representations on the original dataset, we follow \cite{thor2007moma} to extend the gold standard with author matches.

\item IMDB-YAGO (abbr. I-Y) is a large dataset about movies and actors. Following \cite{julien2013sigma}, we generate the gold standard based on ``external links'' in Wikipedia pages. 

\item DBpedia-YAGO (abbr. D-Y) is a large dataset with heterogeneous attributes and relationships. We use the same version as in \cite{zhuang2017hike,suchanek2011paris}. 
\end{itemize}

\begin{table}
\centering
\caption{Statistics of the datasets}
\label{tab:stats}
\begin{tabular}{|c|cccc|}
	\hline 	& \#Entities & \#Attributes & \#Relationships & \#Matches \\
	\hline 	IIMB & 365 / 365 & 12 / 12  & 15 / 15 & 365 \\
			D-A & 2.61K / 64.3K & 3 / 3 & 1 / 1 & 5.35K \\
		         I-Y & 15.1M / 3.04M & 14 / 36 & 15 / 33 & 77K \\
		         D-Y & 3.12M / 3.04M & 684 / 36 & 688 / 33 & 1.31M \\
	\hline 
\end{tabular}
\end{table}

\begin{table}
\centering
\caption{F1-score and number of questions with real workers}
\label{tab:real}
\begin{adjustbox}{width=\columnwidth}
\begin{tabular}{|c|cr|cr|cr|cr|}
	\hline	\multirow{2}{*}{} & \multicolumn{2}{c|}{Remp} & \multicolumn{2}{c|}{HIKE} & \multicolumn{2}{c|}{POWER} & \multicolumn{2}{c|}{Corleone} \\
	\cline{2-9}	& F1 & \#Q & F1 & \#Q & F1 & \#Q & F1 & \#Q \\
	\hline 	IIMB & \textbf{95.3\%} & \textbf{10} & 84.4\% & 70  & 82.4\% & 70 & 94.7\% & 173 \\
			D-A 	& \textbf{97.7\%} & \textbf{60} & 93.3\% & 80 & 94.8\% & 70 & 94.5\% & 161 \\
			I-Y 	& \textbf{70.9\%} & \textbf{110} & 68.1\% & 270 & 69.3\% & 240 & 64.5\% & 402 \\
			D-Y 	& \textbf{87.2\%} & \textbf{130} & 86.4\% & 500 & 84.3\% & 500 & 76.3\% & 1166 \\
	\hline
\end{tabular}
\end{adjustbox}
\end{table}

\noindent\textbf{Competitors.} We compare Remp with three state-of-the-art crowdsourced ER approaches, namely, HIKE \cite{zhuang2017hike}, POWER \cite{chai2018power} and Corleone \cite{gokhale2014corleone}. We have introduced them in Section~\ref{sect:lit}. Since POWER and Corleone are designed for tabular data, we follow HIKE to partition entities into different clusters and deploy POWER and Corleone on each entity cluster. Specifically, IIMB, D-A and I-Y have clear type information, which is directly used to partition entities; for D-Y which does not have clear type information, we reuse the partitioning algorithm presented in HIKE.

\noindent\textbf{Setup.} We implement Remp and all competing methods (as their codes are not available) in Python 3 and C++, and strictly follow each competitor's reported parameters in the respective paper. All our codes are open sourced\footnote{\url{https://github.com/nju-websoft/Remp}}. All experiments are conducted on a workstation with an Intel Xeon 3.3GHz CPU and 128GB RAM. For Remp, we uniformly assign $k = 4$, $\tau = 0.9$ and $\mu = 10$, and use $0.3$ as the label similarity threshold. Similar to~\cite{gokhale2014corleone,zhuang2017hike,chai2018power}, we first prune out all definite non-matches (outlined in Section~\ref{sect:graph}), and all methods take the same retained entity matches $M_{rd}$ as input.

\begin{figure*}
\includegraphics[width=\textwidth]{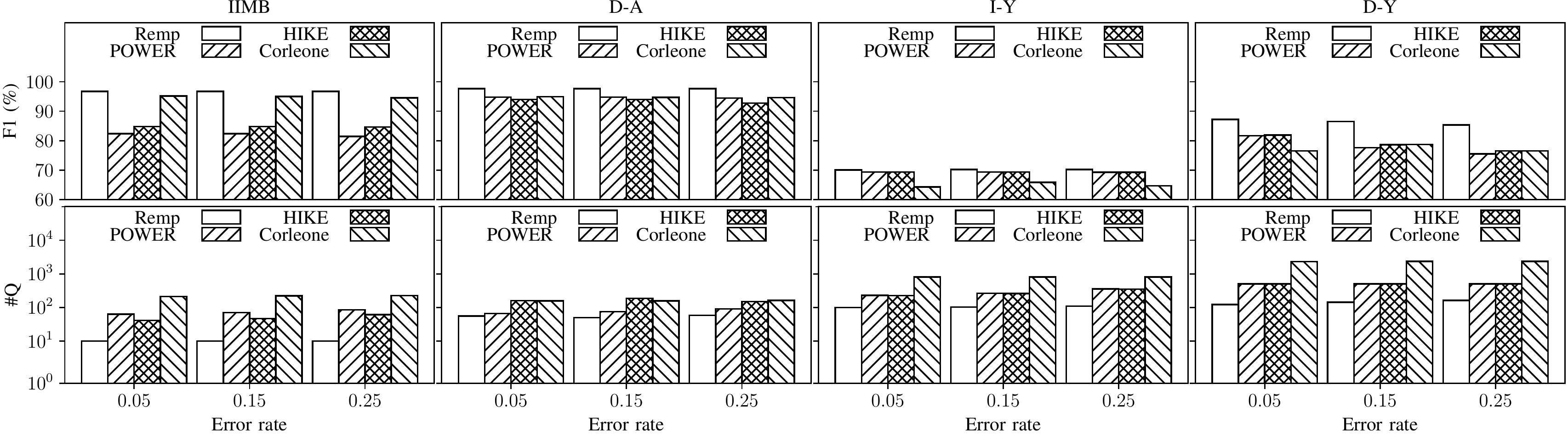}
\caption{F1-score and number of questions w.r.t. simulated workers of varying error rates}
\label{fig:simulated}
\end{figure*}

%--------------------

\subsection{Remp vs. State of the Art}
We set up two experiments, one is with real workers and the other is with simulated workers. The evaluation metrics are the F1-score and the number of questions (\#Q). 

\noindent\textbf{Experiment with real workers.} We publish the questions selected by each approach on Amazon MTurk. Each question is labeled by five workers to decide whether the two entities refer to the same object in the real world. We leverage the common worker qualifications to avoid spammers, i.e. we only allow workers with an approval rate of at least $95\%$. Furthermore, we reuse the label to each question for all approaches. Thus, all approaches can receive the same label to the same question. In total, 651 real workers labeled 3,484 questions.

The results are presented in Table~\ref{tab:real}, and we have the following findings -- (1) Remp consistently achieves the best F1-score with the fewest questions. (2) Remp improves the F1-score moderately, and reduces the number of questions significantly. (3) Specifically, compared with the second best result, Remp reduces the average number of questions by 85.7\%, 14.3\%, 54.2\% and 74.0\% on IIMB, D-A, I-Y and D-Y, respectively. To summarize, Remp achieves the best F1-score and saves the number of questions, especially when the dataset contains various relationships (e.g., D-Y).

\noindent\textbf{Experiment with simulated workers.} We also generate simulated workers who give wrong labels to questions with a fixed probability (called \emph{error rate}). We follow HIKE to set the error rate of simulated workers to 0.05, 0.15 and 0.25.

Figure~\ref{fig:simulated} shows the comparison results and we make several observations -- (1) All approaches obtain stable F1-scores, indicating their robustness in handling imperfect labeling. (2) Remp consistently obtains the highest F1-score, and beats the second best result by 0.4\%, 3.0\%, 1.6\%, 8.0\% on IIMB, D-A, I-Y and D-Y, respectively. This is attributed to Remp's robustness in uncovering matches with low literal similarities as compared to its competitors. For example, literal information is insufficient on I-Y and D-Y, thereby causing errors in the partial order of HIKE and POWER as well as the rules of Corleone. (3) Remp needs considerably fewer questions on IIMB, I-Y and D-Y. Compared with the second best result, Remp reduces the average number of questions by 79.9\%, 26.7\%, 62.5\% and 71.4\% on IIMB, D-A, I-Y and D-Y, respectively. One reason is that there are many types of entities on these datasets, and most matches are linked by relationships of different domain/range types. However, HIKE, POWER and Corleone cannot infer these matches efficiently. (4) On the D-A dataset, Remp only reduces six more questions than POWER, because in the ER graph there are many isolated components but only one type of relationship, making Remp have to check them all.

%--------------------

\subsection{Internal Evaluation of Remp}

In this section, we evaluate how each major module of Remp contributes to its overall performance.

\noindent\textbf{Effectiveness of attribute matching.} For the I-Y dataset, we reuse the gold standard created by SiGMa \cite{julien2013sigma}. For the D-Y dataset, we follow the recommendation of YAGO and extract 19 attribute matches from the subPropertyOf links\footnote{\url{http://webdam.inria.fr/paris/yd_relations.zip}} as the gold standard. Note that it is not necessary to match attributes for the other two datasets. We employ the conventional precision, recall and F1-score as our evaluation metrics. 

As depicted in Table~\ref{tab:attr_match}, Remp performs perfectly on the I-Y dataset but gains a relatively low recall on the D-Y dataset, and the 1:1 matching constraint helps Remp improve the precision. We observe that Remp fails to identify several attribute matches when the attribute pairs rarely appear in $\initialmatches$ (i.e. entity pairs from exact string matching), or when the values are dramatically different (e.g., the $icd10$ value for \textit{dbp:Trigeminal\_neuralgia} is ``G44.847'', but for \textit{yago:Trigeminal\_neuralgia} is ``G-50.0''). We argue that our attribute matches are sufficient to ER, since the first type of missing matches only helps resolve a small portion of entities but increases the running time of building similarity vectors, while the second type requires extra value processing/correction steps before computing the similarities.

\begin{table}
\centering
\caption{Effectiveness of attribute matching}
\label{tab:attr_match}
\begin{adjustbox}{width=\columnwidth}
\begin{tabular}{|c|c|ccc|ccc|}
	\hline \multirow{2}{*}{} & \#Ref. & \multicolumn{3}{|c|}{Remp} & \multicolumn{3}{|c|}{Remp w/o 1:1 matching} \\
	\cline{3-8} & matches & Precision & Recall & F1 & Precision & Recall & F1 \\
	\hline I-Y & \ 4 &  \ 100\% & \ 100\% & \ 100\% & 40.0\% & \ 100\% & 57.1\% \\
		 D-Y & 19 & 90.9\% & 52.6\% & 66.7\% & 52.4\% & 57.9\% & 55.0\% \\
	\hline
\end{tabular}
\end{adjustbox}
\end{table}

\noindent\textbf{Effectiveness of partial order based pruning.} %We first follow the previous work \cite{chai2018power} to prune the entity pairs with similarity smaller than $0.3$ and obtain the candidate matches ($M_c$). Then, 
To test the performance of the entity pair pruning module in Remp, we employ two metrics: (i) the reduction ratio (RR), which is the proportion of pruned candidates, and (ii) the pair completeness (PC), which is the proportion of true matches preserved in candidate/retained matches. We also use the error rate of optimal monotone classifier defined in~\cite{tao2018} to measure the incorrectness of the partial order.

As shown in Table~\ref{tab:pruning}, candidate matches contain most true matches on IIMB, D-A and I-Y, but only $88.7\%$ of true matches on the D-Y dataset. This is because on this dataset 8.4\% of the entities in the true matches lack labels. On IIMB and D-A, Remp has a relatively low RR, because the true matches account for $61.6\%$ and $22.1\%$ of the  candidate matches, respectively. On I-Y and D-Y, the PC of retained matches is close to that of candidate matches, but most candidate matches are pruned. This indicates that the entity pair pruning module is effective. We notice that the error rate on each dataset is nearly perfect, but the other monotonicity-based approaches (i.e., POWER and HIKE) achieve worse accuracy (see Table~\ref{tab:real}). The main reason is that our partial order is restricted to neighbors of each entity pair, where errors do not propagate to the whole candidate match set.

Furthermore, the pair completeness of retained matches w.r.t. varying $k$ is shown in Figure~\ref{fig:pruning_iter}. The pair completeness converges quickly on IIMB, D-A and I-Y but slowly on D-Y, because many matches have only one or two shared attributes, making the partial order work inefficiently. 

\begin{table}
\centering
\caption{Effectiveness of partial order based pruning}
\label{tab:pruning}
\begin{adjustbox}{width=\columnwidth}
\begin{tabular}{|c|rr|rrrr|}
	\hline \multirow{2}{*}{$k = 4$}	& \multicolumn{2}{c|}{Candidate matches} & \multicolumn{4}{c|}{Retained matches} \\ 
	\cline{2-7}     & \#Pairs & PC & \#Pairs (RR) & PC & \#Edges & Error rate \\
	\hline	IIMB    & 593 & 97.8\% & 516 (13.0\%) & 97.8\% & 1K &  1.91\% \\
			D-A 	& 24.2K & 97.9\% & 12.4K (49.0\%) & 97.7\% & 7.6K & 0.37\% \\
			I-Y 	& 2.44B & 98.0\% & 3.86M (99.6\%) & 97.4\% & 0.16M & 0.65\% \\
			D-Y 	& 2.70B & 88.7\% & 13.1M (99.7\%) & 84.8\% & 5.34M & 1.64\%\\
	\hline
\end{tabular}
\end{adjustbox}
\end{table}%
\begin{figure}[t]
\centering
\includegraphics[width=.6\columnwidth]{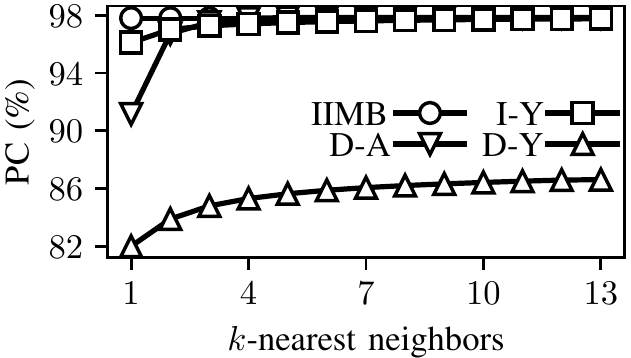}
\caption{Pair completeness w.r.t. $k$-nearest neighbors}
\label{fig:pruning_iter}
\end{figure}%
\begin{figure*}
	\centering
	\includegraphics[width=\textwidth]{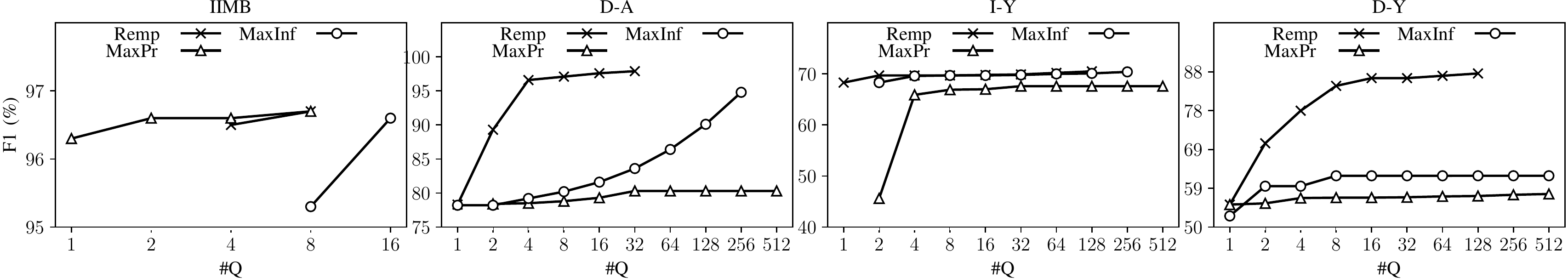}
	\caption{F1-score of Remp, MaxInf and MaxPr w.r.t. varying numbers of questions}
\label{fig:benefit}
\end{figure*}%

\noindent\textbf{Effectiveness of match propagation.} We additionally compare the match propagation module of Remp with two collective, non-crowdsourcing ER approaches: PARIS~\cite{suchanek2011paris} and SiGMa \cite{julien2013sigma}. More details about them have been given in Section~\ref{sect:lit}. To assess the real propagation capability of Remp, we ignore the classifier for handling isolated entity pairs. We randomly sample different portions of entity matches as the seeds for Remp, PARIS and SiGMa. The experiments are repeated five times and the F1-score is reported in Table~\ref{tab:propagation}. We observe that Remp achieves the best F1-score on D-A, I-Y and D-Y. On IIMB (20\% of matches), the F1-score of Remp is slightly worse than that of SiGMa, because SiGMa can obtain matches between isolated entities based on their literal similarities directly. Overall, Remp can achieve the highest F1-score in most cases.

\begin{table}
\centering
\caption{F1-score w.r.t. varying portions of seed matches}
\label{tab:propagation}
\begin{tabular}{|c|l|cccc|}
	\hline  \multicolumn{2}{|c|}{} & \multicolumn{4}{c|}{\% of matches} \\ 
	\cline{3-6}  \multicolumn{2}{|c|}{} & 20 & 40 & 60 & 80 \\
	\hline \multirow{3}{*}{IIMB} 
		& Remp	& 97.5\% & \textbf{98.6\%} & \textbf{99.7\%} & \textbf{99.7\%} \\
		& PARIS	& 96.0\% & 96.5\% & 97.0\% & 97.4\% \\
		& SiGMa	& \textbf{97.6\%} & \textbf{98.6\%} & 99.0\% & 99.6\% \\
	\hline \multirow{3}{*}{D-A}
		& Remp	& \textbf{93.3\%} & \textbf{97.2\%} & \textbf{98.9\%} & \textbf{99.7\%} \\
		& PARIS	& 71.3\% & 79.1\% & 86.2\% & 92.5\% \\
		& SiGMa	& 92.7\% & 94.9\% & 96.7\% & 98.4\% \\
	\hline \multirow{3}{*}{I-Y}
		& Remp	& \textbf{41.2\%} & \textbf{63.4\%} & \textbf{78.8\%} & \textbf{90.6\%} \\
		& PARIS	& 34.8\% & 57.9\% & 75.4\% & 89.0\%  \\
		& SiGMa	& 34.0\% & 58.5\% & 76.1\% & 89.3\% \\
	\hline \multirow{3}{*}{D-Y}
		& Remp	& \textbf{83.2\%} & \textbf{91.4\%} & \textbf{95.0\%} & \textbf{99.7\%} \\
		& PARIS 	& 82.2\% & 84.7\% & 87.2\% & 89.5\% \\
		& SiGMa 	& 33.6\% & 57.4\% & 75.3\% & 89.1\% \\
	\hline
\end{tabular}
\end{table}

\noindent\textbf{Effectiveness of question selection benefit.} We implement two alternative heuristics as baselines, namely MaxInf and MaxPr, to evaluate the question selection benefit. We set $\mu = 1$ and use ground truths as labels. MaxInf selects the questions with the maximal inference power. MaxPr chooses the questions with the maximal match probability. Figure~\ref{fig:benefit} depicts the result and each curve starts when the F1-score is greater than 0. We find (1) Remp always achieves the best F1-score with much less number of questions. (2) MaxPr obtains the lowest F1-score except on the IIMB dataset, because it does not consider how many matches can be inferred by the new question. (3) MaxInf performs worse than Remp, as it often chooses non-matches as the questions, making it find fewer matches than Remp using the same number of questions. This experiment demonstrates that our $\mathtt{benefit}$ function is the most effective one.

\noindent\textbf{Effectiveness of multiple questions selection.} Table~\ref{tab:mu} depicts the F1-score, the number of questions (\#Q) and the number of loops (\#L) of the multiple questions selection module (with ground truth as labels), in term of different question number thresholds per round ($\mu=1,5,10, 20$), and our findings are as follows -- (1) Remp achieves a stable F1-score on all datasets. (2) The number of questions increases when $\mu$ increases, especially when $\mu = 10, 20$. This is probably because Remp always asks $\mu$ questions in one human-machine loop, and it has to ask an extra batch of questions when some questions with large $\benefit$ are labeled as non-matches. Although asking multiple questions in one loop increases the monetary cost, it reduces 75\%--94.1\% number of loops when $\mu=20$.

\noindent\textbf{Effectiveness of inference on isolated entity pairs.} We examine the performance of the random forest classifier in each dataset in the experiments with real workers. As depicted in Table~\ref{tab:isolated}, the classifier achieves poor performance on IIMB and D-A. Due to the tiny proportion of isolcated entity pairs in these two datasets, this is probably caused by occasionality. When the portions of isolated matches increase on I-Y and D-Y, the classifier achieves comparable performance to Remp. This demonstrates that Remp can infer enough matches for resolving the entire dataset even if the ER graph does not cover all candidate matches.

\begin{table}
\centering
\caption{F1-score and number of questions with different question number thresholds per round}
\label{tab:mu}
\begin{adjustbox}{width=\columnwidth}
\begin{tabular}{|c|crr|crr|crr|crr|}
	\hline		\multirow{2}{*}{} & \multicolumn{3}{c|}{$\mu=1$} & \multicolumn{3}{c|}{$\mu=5$} & \multicolumn{3}{c|}{$\mu=10$} & \multicolumn{3}{c|}{$\mu=20$} \\
	\cline{2-13}	& F1 & \#Q & \#L & F1 & \#Q & \#L & F1 & \#Q & \#L & F1 & \#Q & \#L \\
	\hline 		IIMB & 96.7\% & 8 & 8 & 96.7\% & 10 & 2 & 96.7\% & 20 & 2 & 96.9\% & 40 & 2 \\
				D-A  & 97.8\% & 52 & 52 & 97.8\% & 60 & 12 & 97.7\% & 60 & 6 & 97.3\% & 80 & 4\\
				I-Y & 71.4\% & 102 & 102 & 71.3\% & 105 & 21 & 71.3\% & 110 & 11 & 71.4\% & 120 & 6\\
				D-Y & 87.3\% & 127 & 127 & 87.2\% & 135 & 27 & 87.3\% & 140 & 14 & 87.2\% & 160 & 8 \\
	\hline
\end{tabular}
\end{adjustbox}
\end{table}

\begin{table}
\centering
\caption{F1-score of inference on isolated entity pairs}
\label{tab:isolated}
\begin{tabular}{|c|c|cc|}
	\hline 	& Isolated matches & Remp & Random forest \\
	\hline	IIMB & \ \,0.3\% & 95.3\% & \ \,0.0\% \\
			D-A & \ \,0.4\% & 97.7\% & 13.7\% \\
			I-Y & 28.1\% & 70.9\% & 66.3\% \\
			D-Y & 60.4\% & 87.2\% & 84.5\% \\
	\hline
\end{tabular}
\end{table}

\noindent\textbf{Efficiency Analysis.} We run each algorithm three times to record the running time on each of the four datasets. The average running time of Algorithm~\ref{algo:pop} on four datasets is 1s, 8s, 3.9h and 3.6h, the average running time of Algorithm~\ref{algo:imsd} is 0.476s, 6.7s, 109s and 1.07h, and the average running time of Algorithm~\ref{algo:mqs} is 0.128s, 1.27s, 78.5s and 1.25h. We follow the analysis in \cite{das2017falcon} to evaluate the performance of Remp on 25\%, 50\%, 75\% and 100\% of candidate (retained) entity matches $M_c$ ($M_{rd}$) on the D-Y dataset. As depicted in Figure~\ref{fig:efficiency}, the running time of Algorithm~\ref{algo:pop} and Algorithm~\ref{algo:imsd} increase linearly as the number of entity pairs increases. The running time of Algorithm~\ref{algo:mqs} on 25\% and 50\% of retained entity matches are close. This is probably because the sizes of some inferred match sets do not increase significantly.
\begin{figure}[t]
\centering
\includegraphics[width=\columnwidth]{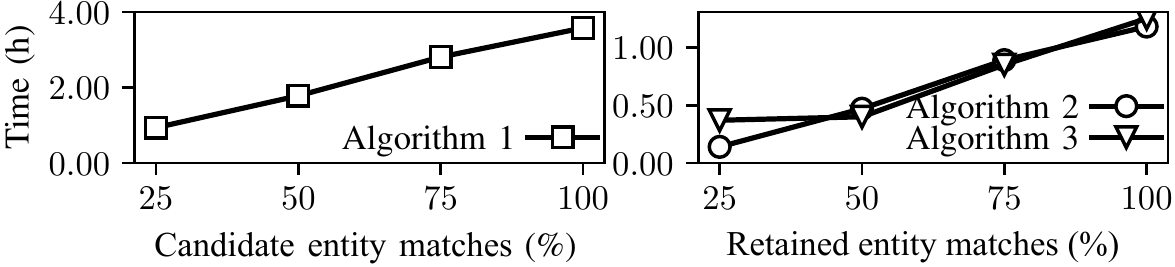}
\caption{Running time w.r.t. different portion of entity pairs}
\label{fig:efficiency}
\end{figure}

%====================

\section{Conclusion}
\label{sect:concl}

In this paper, we proposed a crowdsourced approach leveraging relationships to resolve entities in KBs collectively. Our main contributions are a partial order based pruning algorithm, a relational match propagation model, a constrained multiple questions selection algorithm and an error-tolerant truth inference model. Compared with existing work, our experimental results demonstrated superior ER accuracy and much less number of questions. In future work, we plan to combine transitive relation, partial order and match propagation together as a hybrid ER approach.

%ACKNOWLEDGMENTS are optional
\section*{Acknowledgments}
This work was partially supported by the National Key R\&D Program of China under Grant 2018YFB1004300, the National Natural Science Foundation of China under Grants 61872172, 61772264 and 91646204, and the ARC under Grants DP200102611 and DP180102050. Zhifeng Bao is the recipient of Google Faculty Award.

\balance

\bibliographystyle{IEEEtran}

\bibliography{icde}

\end{document}